\documentclass[12pt,a4paper]{amsart}
\usepackage[all]{xy}
\usepackage{amssymb, amsmath, amscd, geometry}
\usepackage{graphicx}
\usepackage{setspace}

\numberwithin{equation}{section}
\input xy
\xyoption{all}

\usepackage{latexsym}
\usepackage[usenames]{color}
\usepackage{epic} 
\usepackage{mathrsfs}
\usepackage{euscript}
\usepackage{rotating}
\usepackage{verbatim}
\usepackage{todonotes}

\newtheorem{lemma}{Lemma}[section]
\newtheorem{theorem}[lemma]{Theorem}

\newtheorem{prop}[lemma]{Proposition}

\newtheorem*{theorem*}{Theorem}
\newtheorem{remark}[lemma]{Remark}

\theoremstyle{definition}

      \newcommand{\R}{{\mathbb R}}

\title[Sampling quantum nonlocal correlations with high probability]{Sampling quantum nonlocal correlations with high probability}

\addtocounter{tocdepth}{-1}

\begin{document}

\addtolength{\parskip}{+1ex}


\maketitle

\vspace{-0,5em}

C. E. Gonz\'alez-Guill\'en$^{1,4}$, C. H. Jim\'enez$^{2}$, C. Palazuelos$^{3,5}$ and I. Villanueva$^{3,4}$\\

\footnotesize
\vspace{-0,5em}
\noindent$^1$Departamento de Matem\'aticas del \'Area Industrial, E.T.S.I. Industriales, Universidad Polit\'ecnica de Madrid, 28006 Madrid, Spain.\\
$^2$Departamento de Matem\'atica, Pontifícia Universidade Católica do Rio de Janeiro, G\'avea - Rio de Janeiro, 22451900, Brazil.\\
$^3$Departamento de An\'alisis Matem\'atico, Universidad Complutense de Madrid, 28040 Madrid, Spain.\\
$^4$IMI, Universidad Complutense de Madrid, 28040 Madrid, Spain.\\
$^5$Instituto de Ciencias Matem\'aticas (ICMAT), 28049 Madrid. Spain.\\
{\let\thefootnote\relax\footnote{The final publication is available at Springer via http://dx.doi.org/10.1007/s00220-016-2625-8}}
\setcounter{footnote}{0}

\normalsize
\begin{abstract}
It is well known that quantum correlations for bipartite dichotomic measurements are those of the form $\gamma=(\langle u_i,v_j\rangle)_{i,j=1}^n$, where the vectors $u_i$ and $v_j$ are in the unit ball of a real Hilbert space. In this work we study the probability of the nonlocal nature of these correlations as a function of $\alpha=\frac{m}{n}$, where the previous vectors are sampled according to the Haar measure in the unit sphere of $\mathbb R^m$. In particular, we prove the existence of an $\alpha_0>0$ such that if $\alpha\leq \alpha_0$, $\gamma$ is nonlocal with probability tending to $1$ as $n\rightarrow \infty$, while for $\alpha> 2$, $\gamma$ is local with probability tending to $1$ as $n\rightarrow \infty$.
\end{abstract}

\section*{Introduction}
It is well known that local measurements on entangled bipartite quantum states can lead to correlations that cannot be explained by Local Hidden Variable Models (LHVM) \cite{Bell}. This phenomenon, known as \emph{quantum nonlocality}, is one of the most relevant features of quantum mechanics. In fact, though initially discovered in the context of foundations of quantum mechanics, during the last decade quantum nonlocality has become a crucial resource in many applications; some of them are quantum cryptography \cite{ABGMPS,AGM,VaTh1}, communication complexity \cite{BCMW} and random number generators \cite{Pironio,VaTh2}. In this work, we will consider a particularly simple but still relevant context, where two spatially separated observers, Alice and Bob, perform dichotomic (two-outcome) measurements on a bipartite quantum state $\rho$, each on their part of the system. The simplicity of this scenario has made it the natural one  to start developing the previously mentioned applications and also in the experimental verification of the quantum nonlocality phenomenon (see for instance \cite{Aspect,AGR}).

According to the postulates of quantum mechanics, a two-outcome measurement for Alice (resp. Bob) is given by $\{A^+, A^-\}$ (resp. $\{B^+, B^-\}$), where $A^{\pm}$ (resp. $B^{\pm}$) are projectors, or, in general, positive semidefinite operators, acting on a Hilbert space that sum to the identity. We define the observable corresponding to Alice's (Bob's) measurement by $A = A^+ -A^-$ ($B = B^+ -B^-$). The joint correlation of Alice's and Bob's measurement results, denoted by $a$ and $b$ respectively, is $\langle ab \rangle=tr(A\otimes B \rho)$, where $\rho$ is a density matrix\footnote{A density matrix is a positive operator $\rho:H\rightarrow H$ acting on a Hilbert space $H$ with $tr(\rho)=1$.} describing the physical system. Motivated by this, we say that $\gamma=(\gamma_{i,j})_{i,j=1}^n$ is a \emph{quantum correlation matrix}  and denote by $\gamma\in \mathcal Q$, if there exist a density matrix $\rho$ acting on a tensor product of Hilbert spaces $H_1\otimes H_2$ and two families of contractive self-adjoint operators $\{A_i\}_{i=1}^n$, $\{B_i\}_{i=1}^n$  acting on $H_1$ and $H_2$ respectively such that 
\begin{align}\label{QCM1}
\gamma_{i,j}=tr(A_i\otimes B_j \rho)\text{       } \text{    for every    }i,j=1,\dots, n.
\end{align}That is, $\gamma$ is a matrix whose entries are the correlations obtained in an Alice-Bob scenario where each of the observers can choose among $n$ different possible dichotomic measurements. 

We say that $\gamma=(\gamma_{i,j})_{i,j=1}^n$ is a \emph{local correlation matrix} if it belongs to the convex hull
\begin{align}\label{CCM}
\mathcal L=conv\Big\{(\alpha_i\beta_j)_{i,j=1}^n : \alpha_i=\pm 1, \beta_j=\pm 1, \text{   } i,j=1,\dots, n\Big\}.
\end{align}
Local correlation matrices are precisely those whose entries are the correlations obtained in an Alice-Bob scenario when the measurement procedure can be explained by means of a LHVM. It is well known (see \cite{Tsirelson}) that $\mathcal L$ and $\mathcal Q$ are convex sets satisfying $$\mathcal L\varsubsetneq \mathcal Q\varsubsetneq K_G\mathcal L,$$where $1.67696...\leq K_G \leq 1.78221...$ is the so called \emph{Grothendieck's constant}\footnote{The exact value of Grothendieck's constant is still unknown.}. Indeed, the first strict inclusion exactly means that there exist quantum correlations that cannot be explained by means of a LHVM (what we have called quantum nonlocality above) while the second inclusion is a consequence of Grothendieck's inequality (see Theorem \ref{Grothendieck} below) and a result proved by Tsirelson \cite{Tsirelson}, which states that $\gamma=(\gamma_{i,j})_{i,j=1}^n$ is a quantum correlation matrix if and only if there exist a real Hilbert space $H$ and unit vectors $u_1,\dots, u_n, v_1,\dots, v_n$ in $H$ such that
\begin{align}\label{QCM2}
\gamma_{i,j}=\langle u_i,v_j\rangle \text{       } \text{    for every    }i,j=1,\dots, n.
\end{align}

As we have just mentioned, we know of the existence of quantum correlations that are nonlocal. A natural question appears now: how common is nonlocality among quantum correlations? That is, if we pick ``randomly'' a quantum correlation, what is the probability that it is nonlocal? To study this problem, we first need to choose a probability distribution on the set of quantum correlations, in other words, a way of sampling these matrices. We see at least two natural candidates for this. At first sight, it would seem from expression (\ref{QCM1}) that a natural procedure would be sampling on the set of states $\rho$ and on the set of families of self-adjoint and contractive operators $A_1,\dots, A_n,B_1,\dots ,B_n$. The problem with this approach is twofold. First, we do not know a natural probability measure on the set of self-adjoint contractive operators. Second, it seems that we would need to allow for Hilbert spaces of very high dimension\footnote{It is known (\cite{Tsirelson}) that every quantum correlation $\gamma=(\gamma_{i,j})_{i,j=1}^n$ can be written as in (\ref{QCM1}) by using a Hilbert space of dimension exponential in $n$ and, furthermore, such a dimension is required in order to describe the extreme points of $\mathcal Q$.}. 

So, we look for the second candidate: looking at the equivalent reformulation (\ref{QCM2}) of a quantum correlation, we do have a natural sampling procedure: we can sample the  vectors $u_1,\dots, u_n, v_1,\dots, v_n$ independently uniformly distributed on the unit sphere of $\mathbb R^m$. It is well known that this is exactly the same as sampling independent normalized $m$-dimensional gaussian vectors. 

Our results will depend on the relation between the dimension $m$ and the number of measurements $n$.  As we will show later, it is very easy to see that if one fixes any finite $m > 1$, the probability that a quantum correlation matrix $\gamma$  sampled according to the previous procedure is nonlocal tends to one as $n$ tends to infinity. It is also simple to see that if $n$ is fixed and $m$ tends to infinity, then the probability that $\gamma$  is not local converges to $0$.

Our main result says that in the ``constant ratio regime'', where the ratio $\alpha=\frac{m}{n}$ remains constant as $n$ grows, both extreme cases are possible: $\gamma$ will be almost surely local for $\alpha$ big enough, whereas $\gamma$ will be almost surely non local for $\alpha$ small enough. 

Specifically, the main result of our work can be condensed in:

\begin{theorem}\label{Main theorem}
Let $n$ and $m$ be two natural numbers and $\alpha=\frac{m}{n}$. Let us consider $2n$ vectors $u_1,\dots, u_n, v_1,\dots, v_n$ sampled independently according to the Haar measure in the unit sphere of $\mathbb R^m$ and let us denote by $\gamma=(\langle u_i,v_j\rangle)_{i,j=1}^n$ the corresponding quantum correlation matrix. 
\begin{itemize}
\item[a)] If $\alpha\leq \alpha_0\approx 0.004$ then $\gamma$ is nonlocal with probability tending to one as $n$ tends to infinity.
\item[b)] If $\alpha> 2$, then $\gamma$ is local with probability tending to one as $n$ tends to infinity.
\end{itemize}
\end{theorem}
This result shows clearly the need of studying the problem as a function of the parameter $\alpha=\frac{m}{n}$. One possible way to think of this problem is the following: say that we want to sample our vectors on a space of large dimension $m$. In that case, how many vectors $u_1,\dots, u_n, v_1,\dots, v_n$ will we need to sample in order to have nonlocality with high probability? Our results show that $n=\frac{m}{2}$ will be too few vectors, whereas $n=\frac{m}{\alpha_0}$ will be enough. 

There is a considerable gap between $\alpha_0$ and $2$. Our techniques could be refined to slightly increase the bound $\alpha_0$, but they will never reach the relevant case $\alpha_0=1$. From the other side, our proof of part b) suggests that a more clever argument could lead to replace $2$ by $K_G$, but again our present approach does not seem to allow for further improvement. Along these lines, it is plausible that a relation between $\alpha$ and $K_G$ describes interesting behaviors of our correlation matrices. 

It would be interesting to understand the problem for the values $\alpha\in (\alpha_0, 2)$ both by reducing this gap and by studying the existence, or not, of a sharp threshold behaviour of the probability of nonlocality.

Interestingly, we will see below that if one samples normalized vectors whose entries are independent Bernoulli variables, the probability of obtaining a nonlocal correlation matrix is zero, since all of them will be local. This means that, in contrast to many other contexts in random matrix theory, considering gaussian and Bernouilli random variables in our problem leads to completely different conclusions.

In order to prove Theorem \ref{Main theorem} we will use a result previously proved in \cite{GPV} on random matrix theory. Similar techniques were previously used in \cite{Ambainis} to study the dual problem. This problem consists in how likely it is for a random XOR game to have a maximum quantum value strictly bigger than a maximum classical value. In  \cite{Ambainis}, the authors studied the values $\omega^*(A)$ and $\omega(A)$ for random matrices\footnote{Although the authors focused on sign matrices, the same proof works in the case of more general random matrices.} $A=(a_{i,j})_{i,j=1}^n$, where $$\omega^*(A)=\sup\Big\{\sum_{i,j=1}^n a_{i,j}\gamma_{i,j}: \gamma\in \mathcal Q\Big\}\text{     }\text{  and    }\text{     }\omega(A)=\sup\Big\{\sum_{i,j=1}^n a_{i,j}\gamma_{i,j}: \gamma\in \mathcal L\Big\}.$$
They concluded that, for any given $\epsilon>0$, $\omega^*(A)\geq (2-\epsilon) n^{\frac{3}{2}}$ and $\omega(A)\leq 1.6651\dots n^{\frac{3}{2}}$ with probability $1-o(1)$ as $n\rightarrow \infty$ in both cases. This result is the starting point for the proof of our Theorem \ref{Main theorem}. Note that stating ${\omega^*(A)}/{\omega(A)}>1$ for some $A$'s is a reformulation (in a quantitive way) of the fact that $\mathcal L\varsubsetneq \mathcal Q$. The elements $A$'s are usually called \emph{correlation Bell inequalities} (or \emph{XOR-games} in the context of computer science) and the fact that ${\omega^*(A)}/{\omega(A)}>1$ is usually referred to as a \emph{Bell inequality violation}.

The paper is organized as follows. In the first section we briefly introduce some basic results, which will be used throughout the whole paper. The proof of Theorem \ref{Main theorem} is presented in Section \ref{Section: lower bound} and Section \ref{Section: upper bound}. The proof of part a) of the theorem, based on some results on random matrix theory, is given in Section \ref{Section: lower bound}, while Section \ref{Section: upper bound} deals with the proof of part b).
\section{Preliminary results}\label{background}
For completeness and to simplify the reading of the paper, we state in this section the known, or essentially known, previous results which we use throughout the paper.

We will refer to a gaussian vector (matrix) as a random vector whose coordinates are independent standard gaussian random variables in $\mathbb R$, i.e., normal random variables with zero mean and unit variance. 

The following Chernoff-like bound follows from standard bounds of the tails of a standard normal random variable.
\begin{prop}\label{Chernoff}
Let $\{X_i\}_{i=1}^n$ be i.i.d real standard gaussian random variables and $a=(a_1,...,a_n)\in \R^n$. Then, for any $t>1$
$$\Pr\Big(\big|\sum_{i=1}^na_iX_i\big|\geq t\Big) \leq 2e^{-\frac{t^2}{2\|a\|^2}}.$$
\end{prop}

We will also use the well known bounds of the norm of a gaussian vector (see for instance \cite[Corollary 2.3]{Barvinok}).

\begin{prop}\label{concentracion} 
Let $(g_{i,j})_{i,j=1}^{n,m}$ be a gaussian matrix and denote $g_i=(g_{i,j})_{j=1}^m$ for every $i=1,\dots , n$. Then, for every $0<\epsilon<1$ and $i=1,\dots,n$, $$\Pr\Big(\|g_i\|\geq \frac{\sqrt{m}}{\sqrt{1-\epsilon}}\Big)\leq e^{-\frac{\epsilon^2 m}{4}} $$
and
$$\Pr\Big( \|g_i\| \leq \sqrt{m}\sqrt{1-\epsilon}\Big)\leq e^{\frac{-\epsilon^2 m}{4}}.$$
As a consequence we have
\begin{align}\label{normalapprox}
\Pr\Big(\sup_{i=1,\dots, n}\Big\|\frac{g_i}{\|g_i\|}-\frac{g_i}{\sqrt{m}}\Big\|>\epsilon\Big)=
\Pr\Big(\sup_{i=1,\dots, n} \Big|1-\frac{\|g_i\|}{\sqrt{m}}\Big|>\epsilon\Big)\leq 2ne^{-\frac{\epsilon^2m}{4}}.
\end{align}
\end{prop}

\begin{remark}\label{Gaussian-Haar}
As we already mentioned in the Introduction, it is completely equivalent to sample a unit vector $u \in S^{n-1}$ according to the Haar measure $\mu_n$ to sample normalized gaussian vectors $g=\frac{1}{\|(g_1,\dots, g_n)\|}(g_1,\dots, g_n)$. That is, both probability distributions are exactly the same (see \cite[Section 3.3]{Barvinok} for a more complete explanation). In particular, Theorem \ref{Main theorem} can be equivalently stated as it is in Theorem \ref{principal} and Theorem \ref{principal II}. 
\end{remark}
The following proposition can be easily deduced from \cite[Lemma 2.2]{Gupta} and Remark \ref{Gaussian-Haar}.
\begin{prop}\label{concentraciondenormas}
Let $\mu_n$ be the Haar measure in $S^{n-1}$ and let $L\subset \mathbb R^n$ be an $m$-dimensional subspace. For a vector $u \in \mathbb R^n$ denote by $P_L(u)$ its orthogonal projection onto $L$. Then, for any $0 < \rho < 1$ we have
$$\mu_n\left(u \in S^{n-1}: \, \|P_L(u)\|\geq \frac{1}{1-\rho}\sqrt{\frac{m}{n}}\right) \leq e^{-\frac{\rho^2m}{4}},$$
and
$$\mu_n\left(u \in S^{n-1}: \, \|P_L(u)\|\leq (1-\rho)\sqrt{\frac{m}{n}}\right) \leq e^{-\frac{\rho^2m}{4}}.$$
\end{prop}
In this work we will only apply Proposition \ref{concentraciondenormas} in the case when $P_L$ is the projection onto the first $m$ coordinates. In fact, this is equivalent to the general statement due to the rotational invariance of the Haar measure.

We say that a real random $n \times n$ matrix $M$ is bi-orthogonally invariant if the distribution on $M_n(\mathbb R)$ of $M$ is equal to that of $O_1MO_2$ for any orthogonal matrices $O_1$ and $O_2$. It is well known and easy to check that gaussian matrices are bi-orthogonally invariant.

The following result is probably known, but we have not found a reference for it. We write a proof, following the ideas of \cite[Lemma 4.3.10]{Petz2}.
\begin{prop}\label{SVD}
Let $A\in M_n(\mathbb R)$ be an $n \times n$ random matrix in some probability space $(\Xi,\mathbb{P})$. If $A$ is bi-orthogonally invariant then there exist random matrices U and V in $(\Xi,\mathbb{P})$ 
such that \\
(i) $U,V$ follow the Haar distribution in the orthogonal group $\mathcal O(n)$. \\
(ii) $U$ and $V$ are independent.\\
(iii) $U$ and $V$ are the matrices whose columns are respectively the left and right singular vectors associated to the ordered singular values of $A$.
\end{prop}
\begin{proof}
For simplicity, we will assume that the set of matrices with repeated singular values has zero measure (as it happens in the gaussian case, which is the one we will use here). In this case, the singular value decomposition is unique with probability one up to the choice of the sign of the right (or left) singular vectors\footnote{The general case follows by considering the set $\mathcal{V}_i$ of right singular vectors associated to the singular value $s_i$, and  taking random choices of orthonormal vectors in $\mathcal{V}_i$ as the associated columns of the matrix $V$. The measure in $\mathcal{V}_i$ is the one induced by the Haar measure, that is, the measure invariant under unitary transformations of $\mathcal{V}_i$ into itself.}. Let $A$ be a random matrix defined in some space $(\Xi,\mathbb{P})$, and let $A(\xi)=U(\xi)\Sigma(\xi)V^*(\xi)$ be the singular value decomposition of $A(\xi)$ where the singular values of $\Sigma(\xi)$ are ordered in decreasing order and the sign (of the first non zero coordinate) of the right singular vectors are taken at random with probability $1/2$. \footnote{We see $A(\xi)=U(\xi)\Sigma(\xi)V^*(\xi)$ with $U(\xi)=U_0(\xi)S(\zeta)$ and $V^*(\xi)=S(\zeta)V_0^*(\xi)$, where $U_0(\xi)$ and $V_0(\xi)$ are a particular choice of matrices in the singular value decomposition and $S(\zeta)$ is a diagonal random matrix wihose entries are i.i.d. $\pm 1$ valued centered Bernoulli random  variables.}

The random matrices $U$ and $V$ fulfill (iii) by construction. To prove (i) and (ii) it is enough to show that for any $B_1,B_2\subset \mathcal O(n)$ and $\Delta\subset M_n(\mathbb R)$ Borel sets, we have
\begin{equation*}
\mathbb{P}\left(U(\xi)\in B_1,\Sigma(\xi) \in \Delta,V^{*}(\xi)\in B_2\right)=\mu_n(B_1) \mathbb{P}\left(\Sigma(\xi)\in \Delta\right)\mu_n(B_2),
\end{equation*}
where $\mu_n$ is the Haar measure in the orthogonal group.

It follows from the biorthogonal invariance of $A$ that for any two (fixed for now) orthogonal matrices $O_1$ and $O_2$, the random matrix $A'=O_1AO_2$ has the same distribution as $A$. Defining $U'=O_1 U$, ${V'^*}=V^*O_2$, it is clear that $A'=U'\Delta {V'^{*}}$ is a singular value decomposition of $A'$ verifying our requirements. Therefore 
\begin{equation*}
\mathbb{P}\left(U(\xi)\in B_1,\Sigma(\xi) \in \Delta,V^{*}(\xi)\in B_2\right)=
\mathbb{P}\left(O_1U(\xi)\in B_1,\Sigma(\xi) \in \Delta,V^{*}(\xi)O_2\in B_2\right).\end{equation*}

If we now let $O_1$ and $O_2$ be distributed according to the Haar measure $\mu_n$ in different probability spaces $(\Xi',\mathbb{P}')$ and $(\Xi'',\mathbb{P}'')$ respectively, we get 
\begin{align*}
&\mathbb{P}(U(\xi)\in B_1,\Sigma(\xi) \in \Delta,V^{*}(\xi)\in B_2)\\
& =(\mathbb{P}'\otimes \mathbb{P}\otimes \mathbb{P}'')(U'(\xi,\xi')\in B_1,\Sigma(\xi) \in \Delta,V^{'*}(\xi,\xi'')\in B_2) \\
&=\int \int \int \chi_{B_1} (O_1(\xi')U(\xi))\chi_\Delta(\Sigma(\xi))\chi_{B_2}(V^*(\xi)O_2(\xi''))d\mathbb{P}(\xi)d\mathbb{P}'(\xi')d\mathbb{P}''(\xi'')\\
& =\int \left(\int \chi_{B_1} (O_1(\xi')U(\xi))d\mathbb{P}'(\xi')\right)\chi_\Delta(\Sigma(\xi))\left(\int\chi_{B_2}(V^*(\xi)O_2(\xi'')) d\mathbb{P}''(\xi'')\right)d\mathbb{P}(\xi) \\
&= \int \mu_n(B_1) \chi_\Delta(\Sigma(\xi))\mu_n(B_2)d\mathbb{P}(\xi) \\
&=\mu_n(B_1) \mathbb{P}(\Sigma(\xi)\in \Delta)\mu_n(B_2),
\end{align*} 
where the fourth equality follows from the rotational invariance of the Haar measure.
\end{proof}

\begin{remark}\label{muestreo}
We will use later the following easy consequence of  Proposition \ref{SVD}: For every $n\in \mathbb N$ there exists a probability space $\Xi$ with three $n\times n$ random matrices $A, U, V$ defined on it such that  $A$ is a gaussian matrix, $U,V$ are independent and Haar distributed in $\mathcal O(n)$, and for almost every $\xi \in \Xi$, $U(\xi)$ and $V(\xi)$ are the right and left singular values of $A(\xi)$ arranged in decreasing order of the singular values.   
\end{remark}
We will need the Marcenko-Pastur law, describing the distribution of the singular values of random matrices:
\begin{theorem}[Marcenko-Pastur law, \cite{MaPa}]\label{Marcenko-Pastur}
Let $A$ be an $n \times  n$ random matrix whose entries $a_{ij}$ are independent real random variables with zero mean and unit
variance. Let $C\in[0, 2]$. With probability $1-o(1)$, the number of singular values $\lambda$ of $A$ that satisfy $\lambda\geq C\sqrt{n}$ is $(f(C)-o(1))n$, where
$$f(C)=\frac{1}{2\pi}\int_{x=C^2}^4\sqrt{\frac{4}{x}-1}dx.$$Here, we say that $h=h(n)$ is $o(1)$ if and only if $\lim_{n\rightarrow \infty}h(n)=0$.
\end{theorem}
We state for completeness the version of Grothendieck's inequality most useful for our purposes (see \cite[Page 172]{DeFl}).
\begin{theorem}[Grothendieck's inequality]\label{Grothendieck}
There exists a universal constant $K_G$ such that for every natural number $n$ and for every real matrix $(a_{i,j})_{i,j=1}^n$ we have
\begin{align*}
\sup\Big\{\Big|\sum_{i,j=1}^na_{i,j}\langle x_i,y_j\rangle\Big|: x_i,y_j\in B_H\Big\}\leq K_G\sup\Big\{\Big|\sum_{i,j=1}^na_{i,j} s_it_j\Big|: s_i,t_j=\pm 1 \Big\},
\end{align*}where the first supremum runs over elements $x_1,\dots, x_n,y_1,\dots, y_n$ in the unit ball of a real Hilbert space $H$.

The exact value of $K_G$ is still unknown but we have $1.67696...\leq K_G \leq 1.78221...$.
\end{theorem}

Finally, we will introduce two norms that appear naturally in our context. We will use them in Section \ref{Section: upper bound}. We will denote by $\ell_\infty^n\otimes_{\pi} \ell_\infty^n$ the space of $n\times n$ real matrices $M$ endowed with the norm $$\|M\|_{\ell_\infty^n\otimes_{\pi} \ell_\infty^n}=\inf \Big\{\sum_{k=1}^N \lambda_k:\lambda_k\geq 0, M=\sum_{k=1}^N \lambda_k\eta_k\Big\},$$where  $\eta_k$ denotes the matrix associated to a deterministic, therefore local, correlation. That is, for every $k$ we have that $\eta_k=a_k\otimes b_k$ for certain sign vectors $a_k,b_k\in \mathbb R^n$. The reader will immediately realize that given a correlation matrix $\gamma$, we have 
\begin{align}\label{local-projective}
\|\gamma\|_{\ell_\infty^n\otimes_{\pi}\ell_\infty^n}\leq 1 \text{      } \text{    if and only if  $\gamma$ is local. }
\end{align}

In addition, we will denote by $\ell_1^n\otimes_{\epsilon} \ell_1^n$ the space of $n\times n$ real matrices $M$ endowed with the norm $$\|M\|_{\ell_1^n\otimes_{\epsilon} \ell_1^n}=\sup\Big\{\Big|\sum_{i,j=1}^nM_{i,j} s_it_j\Big|: s_i,t_j=\pm 1 \Big\}.$$Now, it becomes completely trivial to see that for a correlation Bell inequality $A$, we have
\begin{align}\label{local-injective}
\|A\|_{\ell_1^n\otimes_{\epsilon} \ell_1^n}=\omega(A).
\end{align}

The tensor product notation for the norms is due to the fact that these are actually tensor norms on $\mathbb R^n\otimes \mathbb R^n$. However, we will not use this property explicitly in this paper. The duality between local correlation matrices and correlation Bell inequalities is mathematically expressed as
\begin{align}\label{duality epsilon-pi}
&\|\gamma\|_{\ell_\infty^n\otimes_{\pi}\ell_\infty^n}=\sup\big\{\big|\sum_{i,j=1}^nA_{i,j}\gamma_{i,j}\big|: \|A\|_{\ell_1^n\otimes_{\epsilon}\ell_1^n}\leq 1\big\}\text{      } \text{    and   }\\\nonumber&\|A\|_{\ell_1^n\otimes_{\epsilon}\ell_1^n}=\sup\big\{\big|\sum_{i,j=1}^nA_{i,j}\gamma_{i,j}\big|: \|\gamma\|_{\ell_\infty^n\otimes_{\pi}\ell_\infty^n}\leq 1\big\}.
\end{align}

\section{A lower bound for $\alpha_0$: Part a) of Theorem \ref{Main theorem}}\label{Section: lower bound}
The following result is implicit in the paper \cite{Ambainis}. It provides an abundance of nonlocal correlations very close to being quantum and it is the starting point of our work.
\begin{prop}\label{Ambainisreinterpretado} Let $U=(u_{i,j})_{i,j=1}^n, V=(v_{i,j})_{i,j=1}^n$ 
be two independent orthogonal random matrices distributed according to the Haar measure on the orthogonal group $\mathcal O(n)$. Let $\alpha\in (0,1)$ and $m=\alpha n$. We also denote $\delta=f^{-1}(\alpha)$, where $f$ is the Marcenko-Pastur densitiy function as in Theorem \ref{Marcenko-Pastur}. Let  $\gamma_{i,j}=\langle \frac{\sqrt{n}}{\sqrt{m}}u_i ,\frac{\sqrt{n}}{\sqrt{m}}v_j \rangle$ with $ u_i=(u_{i,k})_{k=1}^m$ $v_j=(v_{j,k})_{k=1}^m$. Then there exists a random $n\times n$ matrix $A=(a_{i,j})_{i,j=1}^n$ such that, with probability $1-o(1)$, $$\sum_{i,j=1}^n a_{i,j} \gamma_{i,j} \geq (\delta-o(1)) n^{\frac{3}{2}} \text{      }\text{    and   }\text{      } \omega(A)\leq  1.6651\dots n^{\frac{3}{2}}.$$
\end{prop}
\begin{proof}
We consider $A, U, V$ distributed as in Remark \ref{muestreo}. So, $A=U\Sigma V^*$, where $\Sigma$ is the diagonal matrix of the singular values which we may assume to be arranged in decreasing order. Let $\lambda_1,...,\lambda_m$ be the greatest $m$ singular values of $A$. According to our choice of $\delta$, it follows from Theorem \ref{Marcenko-Pastur} that $\lambda_m\geq  (\delta-o(1))\sqrt{n}$ with probability $1-o(1)$. Then, we have 
\begin{align*}\sum_{i,j=1}^n a_{i,j}\Big\langle \frac{\sqrt{n}}{\sqrt{m}}u_i,\frac{\sqrt{n}}{\sqrt{m}}v_j \Big\rangle&=\frac{n}{m} \sum_{i,j=1}^n a_{i,j}\langle u_i  , v_j \rangle=\frac{n}{m} \sum_{k=1}^m \lambda_k\geq (\delta-o(1))n^{\frac{3}{2}}.
\end{align*}
This proves the first inequality of our statement. For the second one, note that $A$ is a gaussian matrix. Then, the result follows exactly as in \cite[Theorem 4]{Ambainis} from the Chernoff bound of Lemma \ref{Chernoff} taking 
$t=\left(2\sqrt{\ln 2}+2\frac{\sqrt{\ln n}}{\sqrt{n}}\right) n^{3/2}$ and applying a union bound argument.
\end{proof}
Now we can state and prove the first part of Theorem \ref{Main theorem}. It states that most correlations will be nonlocal when $m$ is of the order $\alpha_0 n$. The idea of the proof is the following: On the one hand, Proposition \ref{Ambainisreinterpretado} shows that for that order of $m$ the first $m$ rows/columns of two Haar distributed orthogonal matrices generate a nonlocal correlation with high probability. It also provides a gaussian matrix $A$ that certifies this nonlocality. On the other hand, the following theorem shows that the first $m$ columns of a gaussian matrix are ``close'', in an appropriate $\sup$-euclidean norm, to the first $m$ columns of a Haar distributed orthogonal matrix.
\begin{theorem}\cite[Theorem 1.1]{GPV}\label{Theorem Decoupling}
Let $n$ and $m$ be two natural numbers such that $\alpha=\frac{m}{n}\in (0,1)$. Let $G=(g_{i,j})_{i, j=1}^{n,n}$ be a random matrix whose entries are independent real standard gaussian variables. Let $U\in \mathcal O(n)$ be the random matrix obtained applying the Gram-Schmidt orthonormalization procedure to the columns of $G$.  Then $U$ is Haar distributed and
\begin{align}\label{Equation Decoupling}
\Pr \Big(\sup_{i=1,\dots, n}\big\|F_i^m(G-\sqrt{n}U)\big\|> (1+\epsilon)\theta(\alpha)\sqrt{m}\Big)\leq Kne^{C(\epsilon, \alpha)n}, 
\end{align}where $F_i^m(G-\sqrt{n}U)$ is the $i$-$th$ row of the matrix $G-\sqrt{n}U$ truncated to its first $m$ entries, $K$ is a universal positive constant, $C(\epsilon, \alpha)>0$ is a constant depending only on $\epsilon$ and $\alpha$ and
\begin{align*}
\theta(\alpha)&=\sqrt{ 2-\frac{4}{3} \frac{(1-(1 -\alpha)^{3/2})}{\alpha}}.
\end{align*}
\end{theorem}

Finally, Grothendieck's inequality allows us to translate this $\sup$-euclidean closeness into a big value of the correlation $\gamma$ when tested against the witness $A$.

\begin{theorem}\label{principal}
Let $G=(g_{i,j})_{i, j=1}^{n,m}$ and $H=(h_{i,j})_{i, j=1}^{n,m}$ be two random matrices whose entries are independent real standard gaussian variables satisfying $\alpha=\frac{m}{n}\in (0,1)$. For every $i,j=1,\dots , n$, let $g_i=(g_{i,k})_{k=1}^m$ and $h_j=(h_{j,k})_{k=1}^m$ be the row vectors  of $G$ and $H$ respectively . Let us denote  $\bar{g}_i=\frac{g_i}{\|g_i\|}$ and $ \bar{h}_j=\frac{h_j}{\|h_j\|}$. Then, if $\alpha\leq \alpha_0\approx 0.004$, the quantum correlation matrix given by $\gamma=(\langle \bar{g}_i,\bar{h}_j\rangle)_{i,j=1}^n$is not local with probability $1-o(1)$.
\end{theorem} 
As we will explain below, it suffices to show the result for $\alpha=\alpha_0$.

\begin{proof} We consider two independent $n \times n$ gaussian matrices $G'$ and $H'$. We take $G$ and $H$ from the statement as the first $m$ columns of $G'$ and $H'$ respectively. \footnote{Note that the last $n-m$ columns of $G'$ and $H'$ will not play any role in the proof. They are only introduced in order to apply Theorem \ref{Theorem Decoupling} in a simple way.} We can apply Theorem \ref{Theorem Decoupling} to $G'$ (respectively $H'$) to obtain a Haar distributed random matrix $U$ (respectively $V$) such that Equation (\ref{Equation Decoupling}) is fulfilled. Moreover, $U$ and $V$ are independent from each other as so are $G'$ and $H'$. We define $u_i=(u_{i,k})_{k=1}^m$ and $v_j=(v_{j,k})_{k=1}^m$ as the vectors formed by the first $m$ entries of the row vectors of $U$ and $V$ respectively.

It follows from Proposition \ref{Ambainisreinterpretado} that, with probability $1-o(1)$, there exists a gaussian matrix $A$ such that $U, V$ are formed by the left and right singular vectors of $A$, arranged in decreasing order and such that  
\begin{equation}\label{cotaclasica}
\omega(A)
\leq 1.6651\dots n^\frac{3}{2}
\end{equation}
and 
\begin{align*}\sum_{i,j=1}^n a_{i,j} \Big\langle \frac{\sqrt{n}u_i}{\sqrt{m}},\frac{\sqrt{n}v_j}{\sqrt{m}} \Big\rangle \geq (\delta-o(1)) n^{\frac{3}{2}},
\end{align*}
where $\delta=f^{-1}(\alpha)$ as in Proposition \ref{Ambainisreinterpretado}.

We need to see now that $\sum_{i,j=1}^n a_{i,j} \gamma_{i,j}$ is greater than $1.6651\dots n^\frac{3}{2}$. We write
\begin{align*}
\bar{g}_i=\frac{\sqrt{n}u_i}{\sqrt{m}}+ \frac{g_i}{\sqrt{m}}-\frac{\sqrt{n}u_i}{\sqrt{m}}+\bar{g}_i -\frac{g_i}{\sqrt{m}}:= \frac{\sqrt{n}u_i}{\sqrt{m}} + \varepsilon_i,
\end{align*}
\begin{align*}
\bar{h}_j=\frac{\sqrt{n}v_j}{\sqrt{m}}+ \frac{h_j}{\sqrt{m}}-\frac{\sqrt{n}v_j}{\sqrt{m}}+\bar{h}_j -\frac{h_j}{\sqrt{m}}:= \frac{\sqrt{n}v_j}{\sqrt{m}} + \sigma_j.
\end{align*}Therefore, $\Big|\sum_{i,j=1}^n a_{i,j} \gamma_{i,j}\Big|=\Big|\sum_{i,j=1}^n a_{i,j} \langle \bar{g}_i,\bar{h}_j\rangle\Big|$ is lower bounded by 
\begin{align}\label{comienzo}
\Big|\sum_{i,j=1}^n a_{i,j} \Big\langle \frac{\sqrt{n}u_i}{\sqrt{m}},\frac{\sqrt{n}v_j}{\sqrt{m}} \Big\rangle\Big|- \Big|\sum_{i,j=1}^n a_{i,j} \Big\langle \frac{\sqrt{n}u_i}{\sqrt{m}},\sigma_j\Big\rangle\Big|\\\nonumber-\Big|\sum_{i,j=1}^n a_{i,j} \Big\langle \varepsilon_i ,\frac{\sqrt{n}v_j}{\sqrt{m}}\Big\rangle\Big|-\Big|\sum_{i,j=1}^n a_{i,j} \big\langle \varepsilon_i,\sigma_j\big\rangle\Big|.
\end{align}

As we mentioned before, with probability $1-o(1)$ we have 
\begin{align*}\sum_{i,j=1}^n a_{i,j} \Big\langle \frac{\sqrt{n}u_i}{\sqrt{m}},\frac{\sqrt{n}v_j}{\sqrt{m}} \Big\rangle \geq (\delta-o(1)) n^{\frac{3}{2}}.
\end{align*}

We need now to upper bound the other three summands in (\ref{comienzo}). We will do this by means of Theorem \ref{Grothendieck}. First, we need to bound the norm of the vectors $\varepsilon_i, \sigma_j$: We do this for the $\varepsilon_i$'s, since the $\sigma_j$'s are totally analogous.

We note that $$\|\varepsilon_i\|\leq\Big\|\frac{g_i}{\sqrt{m}}-\frac{\sqrt{n}u_i}{\sqrt{m}}\Big\|+\Big\|\bar{g}_i -\frac{g_i}{\sqrt{m}}\Big\|.$$
The second term of this sum can be made arbitrarily small by means of Equation (\ref{normalapprox}).
Moreover, according to Theorem \ref{Theorem Decoupling} we have
\begin{align*}
\Pr \Big(\sup_{i=1,\dots, n}\Big\|\frac{g_i}{\sqrt{m}}-\frac{\sqrt{n}u_i}{\sqrt{m}}\Big\|> (1+\epsilon)\theta(\alpha)\Big)\leq Kne^{C(\epsilon, \alpha)n}.
\end{align*}Thus, for a given $\epsilon>0$ we have that
\begin{align*}
\Pr \Big(\sup_{i=1,\dots, n}\big\|\epsilon_i\big\|> \epsilon+\theta(\alpha)\Big)\leq K'ne^{C'(\epsilon, \alpha)n}.
\end{align*}Also, according to Proposition \ref{concentraciondenormas}  
\begin{align*}
\Pr \Big(\sup_{j=1,\dots, n}\frac{\sqrt{n}}{\sqrt{m}}\|u_j\|> \frac 1 {1-\epsilon}\Big)\leq ne^{-\frac{\epsilon^2m}{4}}.
\end{align*}We have used here that each of the $u_j$'s is the projection onto the first $m$-coordinates of a unit vector distributed according to the Haar measure in $\mathbb R^n$, and we have also used a union bound argument to consider the supremum on $j$. Then, we can invoke Theorem \ref{Grothendieck} to state that for a fixed $\epsilon>0$ we have
\begin{align*}
\Big|\sum_{i,j=1}^n a_{i,j} \Big\langle \varepsilon_i ,\frac{\sqrt{n}v_j}{\sqrt{m}}\Big\rangle\Big|\leq \big(\epsilon+\theta(\alpha)\big)\frac 1 {1-\epsilon}K_G\omega(A)
\end{align*}with probability larger than $1-K''ne^{C'(\epsilon, \alpha)n}$. By following completely analogous arguments for the rest of terms in (\ref{comienzo}) we finally get that 
\begin{align*}
\Big|\sum_{i,j=1}^n a_{i,j} \gamma_{i,j}\Big|\geq (\delta-o(1)) n^{\frac{3}{2}}-\Big(2\big(\epsilon+\theta(\alpha)\big)\frac 1 {1-\epsilon}+\big(\epsilon+\theta(\alpha)\big)^2\Big)K_G\omega(A)
\end{align*}
with probability larger than $1-K_1ne^{-K_2(\epsilon, \alpha)m}$, where $K_1$ is a universal positive constant and $K_2(\epsilon, \alpha)$ is a positive constant depending only on $\epsilon$ and $\alpha$. Then, in order to have a Bell violation and using that $\epsilon$ can be made arbitrarily small, it suffices to impose that 
\begin{align*}
(\delta-o(1)) n^{\frac{3}{2}}>\omega(A)+\big(2\theta(\alpha)+\theta(\alpha)^2\big)K_G\omega(A). 
\end{align*}
It follows from (\ref{cotaclasica}), the relation between $\alpha$ and $\delta$ described in Proposition \ref{Ambainisreinterpretado} and Theorem \ref{Theorem Decoupling} that for $\alpha_0=0.00404$ the previous inequality is verified.
\end{proof}
It is very easy to show that for any fixed finite $m$, the probability that a quantum correlation matrix sampled according to our procedure is nonlocal tends to one as $n$ tends to infinity. We write the proof for the case $m=2$, but the reasoning extends trivially to finite $m$. 

It is well known and easy to check that in the case $n=2$ the element $A=\left(\begin{smallmatrix}1 & 1 \\ 1 & -1\end{smallmatrix}\right)$ verifies $\omega(A)=2$. The fact that $\omega(A)\leq 2$ is usually called the CHSH-inequality. In addition, if we define $u_1=(1,0)$, $u_2=(0,1)$, $v_1=\frac{1}{\sqrt{2}}(1,1)$, $v_2=\frac{1}{\sqrt{2}}(1,-1)$ we have $$\sum_{i,j=1}^2a_{i,j}\langle u_i, v_j\rangle=2\sqrt{2}.$$ Therefore, the quantum correlation matrix given by $\gamma=(\langle u_i, v_j\rangle)_{i,j=1}^2$ is nonlocal. Since the function $$f(u_1,u_2,v_1,v_2)=\sum_{i,j=1}^2a_{i,j}\langle u_i, v_j\rangle$$ is continuous on the cartesian product of unit spheres in $\mathbb R^n$, $S:=S^1\times S^1\times S^1\times S^1$, we can easily conclude the existence of an open subset $B$ of $S$ such that $f_{|B}>2$. By considering a subset of $B$, we can assume that this set has positive Haar measure. Hence, the probability that a quantum correlation matrix sampled according to our procedure with $m=n=2$  is nonlocal is strictly larger than zero. Therefore, if we consider the same sampling as above with $n$ large, we can consider independent $2\times 2$-blocks of $\gamma=(\langle u_i, v_j\rangle)_{i,j=1}^n$ and check the probability that at least one of these blocks is non local. This probability will tend to one as $n$ tends to infinity. 

With the same ideas one can prove that Theorem \ref{principal} remains true if we let  $m<\alpha_0n$. In particular we can also cover the case where $m$ grows sublinearly with $n$.  Call $m_n$ to the dimension we will consider in the case $n$. If $m_n$ stays bounded as $n$ grows to infinity, the reasonings from the previous paragraph apply. Otherwise, for every $n$ we consider only matrices/Bell inequalities $A$ which involve the first $\frac{m_n}{\alpha_0}$ vectors, and apply Theorem \ref{principal}.

To finish this section we will show that sampling normalized vectors $u_i$ and $v_j$ whose entries are independent Bernoulli variables leads to local correlations with probability one. Indeed, if we consider such vectors $u_i=\frac{1}{\sqrt{m}}(\epsilon_1^i,\dots, \epsilon_m^i)$, $v_j=\frac{1}{\sqrt{m}}(\delta_1^j,\dots, \delta_m^j)$, we obtain that $$(\gamma_{i,j})_{i,j=1}^n=\Big(\frac{1}{m}\sum_{k=1}^m\epsilon_k^i\delta_k^j\Big)_{i,j=1}^n.$$However, for a fixed $k$, we have that $(\gamma^k_{i,j})_{i,j=1}^n=\big(\epsilon_k^i\delta_k^j\big)_{i,j=1}^n$ is a deterministic (so local) correlation. Since $(\gamma_{i,j})_{i,j=1}^n$ is written as a convex combination of these objects, we immediately conclude that $(\gamma_{i,j})_{i,j=1}^n$ is a local correlation.
\section{An upper bound for $\alpha_0$: Part b) of Theorem \ref{Main theorem}}\label{Section: upper bound}
\begin{theorem}\label{principal II}
Let $G=(g_{i,j})_{i, j=1}^{n,m}$ and $H=(h_{i,j})_{i, j=1}^{n,m}$ be two random matrices whose entries are independent real standard gaussian variables and let $\alpha=\frac{m}{n}$. For every $i,j=1,\dots , n$, let $g_i=(g_{i,k})_{k=1}^m$ and $h_j=(h_{j,k})_{k=1}^m$ be the %truncated 
row vectors  of $G$ and $H$ respectively . Let us denote  $\bar{g}_i=\frac{g_i}{\|g_i\|}$ and $ \bar{h}_j=\frac{h_j}{\|h_j\|}$. Then, if $\alpha>2$, the quantum correlation matrix given by $\gamma=(\langle \bar{g}_i ,\bar{h}_j\rangle)_{i,j=1}^n$ is local with probability larger than $1-2n^2e^{-C(\alpha)n}$. Here, $C(\alpha)\in (0,1)$ is a constant depending only on $\alpha$.
\end{theorem} 
For the proof of Theorem \ref{principal II} it will be convenient to introduce two new norms. We will denote by $\ell_\infty^n(\ell_2^n)$ the space of $n\times n$ real matrices $M$ endowed with the norm $$\|M\|_{\ell_\infty^n(\ell_2^n)}=\max_{i=1,\dots, n}\Big(\sum_{j=1}^n|M_{i,j}|^2\Big)^{\frac{1}{2}}$$ and $\ell_1^n(\ell_2^n)$ the same space endowed with the norm 
$$\|M\|_{\ell_1^n(\ell_2^n)}=\sum_{i=1}^n\Big(\sum_{j=1}^n|M_{i,j}|^2\Big)^{\frac{1}{2}}.$$ 
An easy application of Cauchy-Schwarz inequality shows that these two norms are dual to each other in the same way as in (\ref{duality epsilon-pi}).

We will also use Khintchine inequality (see for instance \cite[Section 8.5]{DeFl}) that we state next for the particular case we need.
\begin{theorem}[Khintchine inequality]
Let $(\epsilon_i)_{i=1}^n$ be a family of independent Bernoulli variables. Then, for every real numbers $(a_i)_{i=1}^n$ we have that $$\Big (\sum_{i=1}^n|a_i|^2\Big)^{\frac{1}{2}}\leq \sqrt{2}\int_\Omega\Big|\sum_{i=1}^na_i\epsilon_i(\omega)\Big|d\omega.$$
\end{theorem}

The following lemma is the key point in our proof of Theorem \ref{principal II}.
\begin{lemma}\label{theorem 1-summing}
Given an $n\times n$ matrix with real entries $M=(M_{i,j})_{i,j=1}^n$, we have $$\|M\|_{\ell_\infty^n\otimes_{\pi}\ell_\infty^n}\leq \sqrt{2} \|M\|_{\ell_\infty^n(\ell_2^n)}.$$
\end{lemma}
Lemma \ref{theorem 1-summing} is a reformulation of the well known fact that $\pi_1(id:\ell_1^n\rightarrow \ell_2^n)\leq \sqrt{2}$, where here $\pi_1$ denotes the $1$-summing norm. We present here a self-contained proof.
\begin{proof}
By the duality relation (\ref{duality epsilon-pi}) and the comments above, the statement is equivalent to proving that for every real matrix $M=(M_{i,j})_{i,j=1}^n$ we have  $$\|M\|_{ \ell_1^n(\ell_2^n)}\leq\sqrt{2} \|M\|_{\ell_1^n\otimes_\epsilon \ell_1^n}.$$ In order to prove this, we apply Khintchine inequality:
\begin{align*}
 \|M\|_{ \ell_1^n(\ell_2^n)}&=\sum_{i=1}^n\Big(\sum_{j=1}^n|M_{i,j}|^2\Big)^\frac{1}{2}\leq \sqrt{2}\sum_{i=1}^n\int_\Omega\Big|\sum_{j=1}^nM_{i,j}\epsilon_j(\omega)\Big|d\omega\\&=\sqrt{2}\int_\Omega\sum_{i=1}^n\Big|\sum_{j=1}^nM_{i,j}\epsilon_j(\omega)\Big|d\omega\leq \sqrt{2}\int_\Omega \|M\|_{\ell_1^n\otimes_\epsilon \ell_1^n} d\omega= \sqrt{2} \|M\|_{\ell_1^n\otimes_\epsilon \ell_1^n},\end{align*} where the second inequality follows from the easy fact that
$$\|M\|_{\ell_1^n\otimes_\epsilon \ell_1^n}=\sup\Big\{\sum_{i=1}^n\Big|\sum_{j=1}^nM_{i,j}\epsilon_j\Big|:\epsilon_j=\pm 1\Big\}.$$

\end{proof}

\begin{proof}[Proof of Theorem \ref{principal II}]

Let $\gamma=(\langle \bar{g}_i ,\bar{h}_j\rangle)_{i,j=1}^n$ and let us fix $i=1,\dots,n$. Then, 
\begin{align*}
\Big(\sum_{j=1}^n\langle \bar{g}_i, \bar{h}_j\rangle^2\Big)^{\frac{1}{2}}\geq \frac{1}{\min\limits_{j=1,\dots, n} \|h_j\|}\Big(\sum_{j=1}^n\langle \bar{g}_i, h_j\rangle^2\Big)^{\frac{1}{2}}=\frac{1}{\min\limits_{j=1, \dots, n} \|h_j\|}\Big(\sum_{j=1}^n y_{i,j}^2\Big)^{\frac{1}{2}}.
\end{align*}
Now, $\bar{g}_i$ is a unit vector which is independent of the gaussian vectors $h_j$ and then $\left(y_{i,j}\right)_{j=1}^n$ is a gaussian vector. Using Proposition \ref{concentracion} together with a union bound argument, one can easily deduce that  for every $0<\epsilon<1$
\begin{equation*}
\Pr \Big(\frac{1}{\min\limits_{j=1,\dots, n} \|h_j\|}\Big(\sum_{j=1}^n y_{i,j}^2\Big)^{\frac{1}{2}}\geq \frac{1}{1-\epsilon}\sqrt{\frac{n}{m}}\Big)\leq ne^{-\frac{m\epsilon^2}{4}}+e^{-\frac{n\epsilon^2}{4}}.
\end{equation*}

Therefore, if $\frac{m}{n}=\alpha> 2$ we can find an $\epsilon_0=\epsilon(\alpha)\in (0,1)$ so that
\begin{equation*}
\Pr \Big(\Big(\sum_{j=1}^n\langle \bar{g}_i, \bar{h}_j\rangle^2\Big)^{\frac{1}{2}}\geq (1-\epsilon_0)\frac{1}{\sqrt{2}}\Big)\leq 2ne^{-nC(\alpha)},
\end{equation*}where $C(\alpha)$ is a positive constant depending on $\alpha$. Moreover, by a union bound argument we obtain
\begin{align*}
\Pr \Big(\|\gamma\|_{\ell_\infty^n(\ell_2^n)}\geq (1-\epsilon_0)\frac{1}{\sqrt{2}}\Big)&=\Pr \Big(\sup_{i=1,\dots, n}\Big(\sum_{j=1}^n\langle \bar{g}_i, \bar{h}_j\rangle^2\Big)^{\frac{1}{2}}\geq (1-\epsilon_0)\frac{1}{\sqrt{2}}\Big)\\&\leq 2n^2 e^{-nC(\alpha)}.
\end{align*}

According to Lemma \ref{theorem 1-summing}, this shows that 
\begin{equation*}
\Pr \Big(\|\gamma\|_{\ell_\infty^n\otimes_\pi\ell_\infty^n}\geq (1-\epsilon_0)\Big)\leq 2n^2e^{-nC(\alpha)}.
\end{equation*}

By (\ref{local-projective}) the previous equation means that the matrix correlation $\gamma$ is local with probability larger than or equal to $1-2n^2e^{-nC(\alpha)}$. This finishes the proof.
\end{proof}

\section*{Acknowledgments}
We would like to thank Jop Briet and  Yeong-Cherng Liang for helpful discussions on previous versions. The author's research was supported by MINECO (grants MTM2011-26912, MTM2012-30748, MTM2014-54240-P), Comunidad de Madrid (grant QUITEMAD+-CM, ref. S2013/ICE-2801), Spanish ``Ram\'on y Cajal'' program and ICMAT Severo Ochoa project SEV-2015-0554 (MINECO), and Mexico's National Council for Science and Technology postdoctoral grant 180486.
%
%

%\appendix

\end{document}